\documentclass[11pt,letterpaper]{article}

\usepackage[margin=1in]{geometry}
\usepackage{amsmath,amssymb,amsthm,mathtools}
\usepackage[T1]{fontenc}
\usepackage{times}
\usepackage{microtype}
\usepackage{authblk}
\usepackage{xcolor}
\usepackage[colorlinks=true,linkcolor=blue!55!black,citecolor=blue!55!black,urlcolor=blue!55!black]{hyperref}
\hypersetup{
  pdftitle={Shor's Conjecture Is True: Projective Measurements Suffice for Binary Accessible Information},
  pdfauthor={Sunghyeon Jo}
}

\allowdisplaybreaks

\newtheorem{theorem}{Theorem}[section]
\newtheorem{corollary}[theorem]{Corollary}
\newtheorem{proposition}[theorem]{Proposition}
\newtheorem{lemma}[theorem]{Lemma}

\newcommand{\Tr}{\operatorname{Tr}}
\newcommand{\Id}{I}
\newcommand{\E}{\mathcal{E}}
\newcommand{\Iacc}{I_{\mathrm{acc}}}
\newcommand{\cH}{\mathcal{H}}

\title{\Large\bfseries Shor's Conjecture Is True: Projective Measurements Suffice for Binary Accessible Information}
\author{Sunghyeon Jo\thanks{\href{mailto:sjo65@gatech.edu}{\textcolor{black}{sjo65@gatech.edu}}}}
\affil{Georgia Institute of Technology}
\date{}

\begin{document}
\maketitle

\begin{abstract}
Shor conjectured that a von Neumann measurement attains the accessible information of every binary quantum ensemble. We prove the conjecture constructively in arbitrary finite dimension. For every finite-outcome positive operator-valued measure (POVM) \(M\), we form an operator \(T_M\) from the posterior label probabilities and show that its spectral projection-valued measure (PVM) \(\Pi_M\) satisfies \(I_{\Pi_M}(X{:}Y)\ge I_M(X{:}Y)\); every rank-one refinement retains the inequality. Two applications of Jensen's operator inequality prove the comparison and yield an exact concave variational formula for the accessible information. The result is a special case of the general theorem of Fang, Fawzi, and Fawzi on measured \(f\)-divergences; the proof below isolates the binary argument and makes the replacement \(M\mapsto T_M\mapsto\Pi_M\) explicit.
\end{abstract}

\section{Introduction}

The accessible information of a quantum ensemble is the largest Shannon mutual information between the classical label of the prepared state and the outcome of a quantum measurement. Shor conjectured that, for an ensemble of two states, this optimum is always attained by a von Neumann measurement \cite{Shor2000}. We use the standard finite-dimensional model of a von Neumann measurement as a projection-valued measure (PVM). Shor allowed higher-rank projections, but refining a PVM to rank one cannot decrease mutual information, so it is equivalent here to optimize over measurements in orthonormal bases.

Keil proved the conjecture for qubit ensembles \cite{Keil}. A December 2025 preprint by Thai and Dall'Arno described the general finite-dimensional case as open \cite{ThaiDallArno}. Theorem~2 of Fang, Fawzi, and Fawzi, first posted in February 2025, implies the conjecture after binary mutual information is identified with the measured \(f_p\)-divergence described in Section~5 \cite{FFF}. We make this consequence explicit and give a short self-contained constructive proof.

Given a POVM \(M=\{M_y\}\), the proof forms \(T_M\) by weighting each effect \(M_y\) by the posterior probability \(\Pr(X=1\mid Y=y)\), with a fixed convention for outcomes of probability zero. The spectral PVM of \(T_M\) has at least as much mutual information as \(M\). The proof applies Jensen's operator inequality twice and reduces the comparison to a scalar maximization. An affine endpoint regularization handles singular states without changing the spectral projections of \(T_M\).

All Hilbert spaces are finite dimensional, and all logarithms are natural. Division by \(\log 2\) gives information in bits. We state the theorem for finite-outcome POVMs. This entails no loss for accessible information: Davies' theorem guarantees an optimal POVM with at most \(d^2\) outcomes in dimension \(d\) \cite{Davies1978}.

\section{Comparison lemmas}

Let \(\cH\) be a complex Hilbert space of dimension \(d\), and let
\begin{equation}
  \E=\{(p,\rho_1),(q,\rho_0)\},
  \qquad q=1-p,\qquad 0<p<1,
  \label{eq:ensemble}
\end{equation}
where \(\rho_0\) and \(\rho_1\) are density operators on \(\cH\). A finite-outcome POVM is a family \(M=\{M_y\}_{y\in\mathcal Y}\) of positive semidefinite operators such that \(\sum_yM_y=\Id\). A PVM, or projective measurement, is a POVM whose elements are pairwise orthogonal projections. Set
\begin{equation}
  a_y=\Tr(\rho_1M_y),\qquad
  b_y=\Tr(\rho_0M_y),\qquad
  r_y=pa_y+qb_y.
  \label{eq:abr}
\end{equation}
For each outcome, define
\begin{equation}
  \eta_y
  :=
  \begin{cases}
    \dfrac{pa_y}{r_y},&r_y>0,\\[5pt]
    p,&r_y=0,
  \end{cases}
  \qquad
  T_M:=\sum_y\eta_yM_y.
  \label{eq:T-eta}
\end{equation}
For \(r_y>0\), \(\eta_y=\Pr(X=1\mid Y=y)\). If \(r_y=0\), then
\(a_y=b_y=0\), and the second line of \eqref{eq:T-eta} fixes a
convention for an outcome of probability zero. In particular,
\begin{equation}
  0\preceq T_M\preceq\Id.
  \label{eq:T-bounds}
\end{equation}
Let \(X\in\{0,1\}\) be the state label, with \(\Pr(X=1)=p\), and let \(Y\) be the measurement outcome. The induced mutual information is
\begin{equation}
  I_M(X{:}Y)
  =
  \sum_y
  \left[
    pa_y\log\frac{a_y}{r_y}
    +qb_y\log\frac{b_y}{r_y}
  \right],
  \label{eq:MI}
\end{equation}
with the usual continuous conventions at zero.
We also write \(I_M(\E)\) for this quantity when the ensemble must be displayed.

\begin{lemma}[Jensen's operator inequality]
\label{lem:Jensen}
Let \(M=\{M_y\}_{y=1}^m\) be a POVM and let \(c_1,\ldots,c_m>0\). Then
\begin{equation}
  \sum_{y=1}^m(\log c_y)M_y
  \preceq
  \log\!\left(\sum_{y=1}^m c_yM_y\right).
  \label{eq:Jensen}
\end{equation}
\end{lemma}

\begin{proof}
Define \(V:\cH\to\cH\otimes\mathbb C^m\) by
\[
  V\xi=\sum_{y=1}^m M_y^{1/2}\xi\otimes e_y
\]
and \(C=\sum_yc_y\,\Id\otimes|e_y\rangle\langle e_y|\). Then \(V\) is an isometry,
\[
  V^*CV=\sum_yc_yM_y,
  \qquad
  V^*(\log C)V=\sum_y(\log c_y)M_y.
\]
Since \(\log\) is operator concave, Jensen's operator inequality gives
\(V^*(\log C)V\preceq\log(V^*CV)\) \cite{HansenPedersen}.
\end{proof}

\begin{lemma}[Scalar variational identity]
\label{lem:scalar}
For \(A,B\ge0\), define
\begin{equation}
  J_p(A,B)
  :=
  pA\log\frac{A}{pA+qB}
  +qB\log\frac{B}{pA+qB},
  \label{eq:JAB}
\end{equation}
with the usual continuous conventions. Then
\begin{equation}
  J_p(A,B)
  =
  \sup_{0<t<1}
  \left\{
    pA\log\frac{t}{p}
    +qB\log\frac{1-t}{q}
  \right\}.
  \label{eq:scalar-variational}
\end{equation}
If \(A,B>0\), the unique maximizer is
\begin{equation}
  t^*=\frac{pA}{pA+qB}.
  \label{eq:tstar}
\end{equation}
\end{lemma}

\begin{proof}
For \(A,B>0\), the objective in \eqref{eq:scalar-variational} is strictly concave and has derivative
\[
  \frac{pA}{t}-\frac{qB}{1-t}.
\]
Its unique critical point is \eqref{eq:tstar}, and substitution gives \eqref{eq:JAB}. If \(A=0<B\) or \(B=0<A\), the supremum is approached as \(t\downarrow0\) or \(t\uparrow1\), respectively. The case \(A=B=0\) is immediate.
\end{proof}

For \(0\prec T\prec\Id\), define
\begin{equation}
  \Phi_p(T)
  :=
  p\Tr(\rho_1\log T)
  +q\Tr\!\bigl(\rho_0\log(\Id-T)\bigr)
  -p\log p-q\log q.
  \label{eq:Phi}
\end{equation}

\begin{lemma}
\label{lem:spectral}
Let \(0\prec T\prec\Id\), and write its spectral decomposition as
\begin{equation}
  T=\sum_i t_i\Pi_i.
  \label{eq:spectral-decomp}
\end{equation}
Then the PVM \(\Pi=\{\Pi_i\}\) formed by the spectral projections of \(T\) satisfies
\begin{equation}
  I_\Pi(X{:}Y)\ge\Phi_p(T).
  \label{eq:spectral-comparison}
\end{equation}
\end{lemma}

\begin{proof}
Set
\[
  A_i=\Tr(\rho_1\Pi_i),
  \qquad
  B_i=\Tr(\rho_0\Pi_i).
\]
By spectral calculus and Lemma~\ref{lem:scalar},
\begin{align}
  \Phi_p(T)
  &=
  \sum_i
  \left[
    pA_i\log\frac{t_i}{p}
    +qB_i\log\frac{1-t_i}{q}
  \right] \\
  &\le
  \sum_iJ_p(A_i,B_i)
  =
  I_\Pi(X{:}Y).
\end{align}
\end{proof}

\section{Projective measurements suffice}

\begin{theorem}
\label{thm:main}
For every binary ensemble \eqref{eq:ensemble} and every finite-outcome POVM \(M\), the PVM \(\Pi_M\) formed by the spectral projections of \(T_M\) satisfies
\begin{equation}
  I_{\Pi_M}(X{:}Y)\ge I_M(X{:}Y).
  \label{eq:main-comparison}
\end{equation}
The same inequality holds for every rank-one refinement of \(\Pi_M\).
\end{theorem}

\begin{proof}
For \(0<\delta<1/2\), set
\begin{equation}
  \eta_y^{(\delta)}=(1-2\delta)\eta_y+\delta,
  \qquad
  T_\delta
  =\sum_y\eta_y^{(\delta)}M_y
  =(1-2\delta)T_M+\delta\Id.
  \label{eq:T-delta}
\end{equation}
\[
  \delta\Id\preceq T_\delta\preceq(1-\delta)\Id,
\]
so \(0\prec T_\delta\prec\Id\). The strictly increasing affine map
\(t\mapsto(1-2\delta)t+\delta\) preserves all eigenspaces. Thus
\(T_\delta\) and \(T_M\) have the same spectral PVM \(\Pi_M\).

Apply Lemma~\ref{lem:Jensen} to \(c_y=\eta_y^{(\delta)}\) and
\(c_y=1-\eta_y^{(\delta)}\), and take traces against \(p\rho_1\)
and \(q\rho_0\), respectively. Together with Lemma~\ref{lem:spectral},
this gives
\begin{align}
  I_{\Pi_M}(X{:}Y)
  &\ge \Phi_p(T_\delta) \\
  &\ge
  p\sum_y a_y\log\frac{\eta_y^{(\delta)}}p
  +q\sum_y b_y\log\frac{1-\eta_y^{(\delta)}}q
  =:L_\delta.
  \label{eq:L-delta}
\end{align}
If \(r_y>0\), then
\[
  \frac{\eta_y}{p}=\frac{a_y}{r_y},
  \qquad
  \frac{1-\eta_y}{q}=\frac{b_y}{r_y}.
\]
If \(r_y=0\), then \(a_y=b_y=0\). Since the outcome set is finite,
the continuous conventions at zero give
\[
  \lim_{\delta\downarrow0}L_\delta=I_M(X{:}Y).
\]
This proves \eqref{eq:main-comparison}. A rank-one refinement cannot
decrease mutual information by the data-processing inequality.
\end{proof}

\begin{corollary}[Shor's conjecture]
\label{cor:Shor}
For every binary ensemble on a finite-dimensional Hilbert space,
\begin{equation}
  \Iacc(\E)
  =
  \sup_{M\ {\rm POVM}}I_M(X{:}Y)
  =
  \max_{\substack{\Pi\ {\rm rank\mbox{-}one}\\{\rm PVM}}}I_\Pi(X{:}Y).
  \label{eq:optima}
\end{equation}
\end{corollary}

\begin{proof}
Theorem~\ref{thm:main} gives one inequality; the reverse follows because every PVM is a POVM. The set of rank-one PVMs with labeled outcomes is the continuous image of the compact group \(U(d)\), hence compact. Since \(I_\Pi(X{:}Y)\) is continuous in the projectors, the PVM optimum is attained. The standard finite-outcome reduction identifies the POVM supremum with the accessible information \cite{Davies1978}.
\end{proof}

\section{Variational formula}

\begin{proposition}
\label{prop:variational}
For the binary ensemble \eqref{eq:ensemble},
\begin{equation}
  \Iacc(\E)
  =
  \sup_{\,0\prec T\prec\Id}
  \left\{
    p\Tr(\rho_1\log T)
    +q\Tr\!\bigl(\rho_0\log(\Id-T)\bigr)
    -p\log p-q\log q
  \right\}.
  \label{eq:variational}
\end{equation}
If \(\rho_0,\rho_1\succ0\), the supremum is attained for some \(0\prec T\prec\Id\). For singular states, the supremum need not be attained in this domain. If the supremum is attained at \(T_*\), every rank-one refinement of the PVM formed by its spectral projections is optimal.
\end{proposition}

\begin{proof}
For every \(0\prec T\prec\Id\), Lemma~\ref{lem:spectral} gives
\[
  \Phi_p(T)\le I_\Pi(X{:}Y)\le\Iacc(\E),
\]
where \(\Pi\) is formed by the spectral projections of \(T\). Thus the right-hand side of \eqref{eq:variational} is at most \(\Iacc(\E)\).

Conversely, let \(\Pi=\{\Pi_i\}\) be an optimal rank-one PVM, whose existence follows from Corollary~\ref{cor:Shor}, and set
\[
  A_i=\Tr(\rho_1\Pi_i),
  \qquad
  B_i=\Tr(\rho_0\Pi_i).
\]
For \(pA_i+qB_i>0\), let
\[
  t_i^*=\frac{pA_i}{pA_i+qB_i}\in[0,1].
\]
If \(A_i=B_i=0\), choose any \(t_i^*\in(0,1)\). Choose \(t_i^{(n)}\in(0,1)\) with \(t_i^{(n)}\to t_i^*\), and set
\[
  T_n=\sum_i t_i^{(n)}\Pi_i.
\]
Lemma~\ref{lem:scalar} gives
\[
  \Phi_p(T_n)
  \longrightarrow
  \sum_iJ_p(A_i,B_i)
  =
  I_\Pi(X{:}Y)
  =
  \Iacc(\E),
\]
which proves the reverse inequality.

If both states are full rank, then \(A_i,B_i>0\) for every \(i\), so
\[
  T_*=\sum_i\frac{pA_i}{pA_i+qB_i}\Pi_i
\]
lies in the domain \(0\prec T\prec\Id\) and attains the supremum.

Nonattainment can occur for singular states. Let
\(P_1=|1\rangle\langle1|\) and \(P_0=|0\rangle\langle0|\), where
\(\langle0|1\rangle=0\), and take \(\rho_1=P_1\), \(\rho_0=P_0\).
Measuring in a basis containing \(|0\rangle\) and \(|1\rangle\)
reveals the label, so their accessible information is
\[
  h(p):=-p\log p-q\log q.
\]
For every \(0\prec T\prec\Id\), both \(\log T\) and
\(\log(\Id-T)\) are negative definite, so \(\Phi_p(T)<h(p)\).
On the other hand, for \(0<\varepsilon<1/2\),
\[
  T_\varepsilon
  =(1-\varepsilon)P_1+\varepsilon P_0
  +\frac12(\Id-P_0-P_1)
\]
satisfies \(0\prec T_\varepsilon\prec\Id\) and
\[
  \Phi_p(T_\varepsilon)
  =h(p)+\log(1-\varepsilon)
  \longrightarrow h(p).
\]
Thus the supremum in \eqref{eq:variational} is not attained for this
ensemble.

Finally, if the supremum is attained at \(T_*\), Lemma~\ref{lem:spectral} shows that the PVM formed by its spectral projections, and hence every rank-one refinement, is optimal.
\end{proof}

The objective in \eqref{eq:variational} is concave on the convex domain \(0\prec T\prec\Id\).

\section{Relation to measured \texorpdfstring{\(f\)}{f}-divergences}

Define
\[
  f_p(t)
  =
  \begin{cases}
    pt\log t-(pt+q)\log(pt+q),&t\ge0,\\
    +\infty,&t<0,
  \end{cases}
\]
where \(0\log0=0\). For \(t>0\),
\[
  f_p''(t)=\frac{pq}{t(pt+q)}>0.
\]
Moreover, \(\lim_{t\downarrow0}f_p(t)=-q\log q=f_p(0)\), so the
extension by \(+\infty\) to \(t<0\) is convex and lower semicontinuous.
The mutual information generated by a measurement \(M\) is the
classical \(f_p\)-divergence between the two conditional outcome
distributions:
\[
  I_M(X{:}Y)
  =
  \sum_y b_y f_p\!\left(\frac{a_y}{b_y}\right).
\]
Here terms with \(b_y=0\) are defined by the corresponding limit, and \(0f_p(0/0):=0\).
Related sufficient conditions for equality of measured and projectively measured \(f\)-divergences were given earlier by Hiai \cite[Theorem~5.8]{Hiai2021}. The reparametrized criterion of Fang, Fawzi, and Fawzi is the form used here.
Let \(f_p^*\) be the Fenchel conjugate of \(f_p\) and set
\(\psi(t)=p\log(t/p)\) for \(0<t<1\). A direct calculation gives
\[
  \operatorname{dom}f_p^*
  =
  (-\infty,-p\log p)
\]
and
\[
  f_p^*(z)
  =
  q\log q-q\log\!\left(1-pe^{z/p}\right),
  \qquad z<-p\log p.
\]
Thus \(\psi\) maps \((0,1)\) bijectively onto \(\operatorname{dom}f_p^*\), \(\psi\) is operator concave, and
\[
  (f_p^*\circ\psi)(t)
  =
  q\log q-q\log(1-t).
\]
The last function is operator convex. The hypotheses of Theorem 2 of Fang, Fawzi, and Fawzi are therefore satisfied \cite{FFF}. Their theorem gives equality of the POVM and PVM suprema; rank-one refinement and compactness give the maximum in Corollary \ref{cor:Shor}. Specializing their proof to \(f_p\) gives the same two applications of Jensen's operator inequality used above.

\section*{AI Usage Disclosure}

GPT-5.6 Pro was used in exploratory discussions during the development of the proof and to assist in drafting a preliminary version of the manuscript. In particular, the central operator construction and a nearly complete initial proof arose from those discussions. The author checked and revised the argument and takes full responsibility for the paper's contents and correctness.

\end{document}